\documentclass[a4paper,11pt,reqno]{amsart}

\usepackage[headings]{fullpage} 
\usepackage[utf8]{inputenc}
\usepackage{mathrsfs}
\usepackage{dsfont}
\usepackage{hyperref}
\hypersetup{bookmarksdepth=3}
\usepackage{amsmath}

\allowdisplaybreaks[4]
\usepackage{amssymb}
\usepackage{amsthm}
\usepackage{amsfonts}
\usepackage{amstext}
\usepackage{amsopn}
\usepackage{amsxtra}
\usepackage{mathrsfs}
\usepackage{dsfont}
\usepackage{esint}
\usepackage{enumitem}
\usepackage[capitalize]{cleveref}
\usepackage{braket}
\usepackage{mathtools}

\usepackage{subcaption}

\usepackage{tikz}
\usetikzlibrary{arrows}
\usepackage{tkz-tab}
\usetikzlibrary{math}
\usepackage{pgfplots}
\pgfplotsset{compat=1.18}

\usepackage[dvipsnames]{xcolor}
 
\usepackage{csquotes}
\usepackage[backend=bibtex,sorting=nty,citestyle=numeric-comp,bibstyle=ieee,maxbibnames=99,doi=false,isbn=false,url=false,eprint=false,dashed=true]{biblatex}
\AtEveryBibitem{\clearfield{month}}
\AtEveryBibitem{\clearfield{day}}
\theoremstyle{plain}
\newtheorem{theorem}{Theorem}

\newtheorem{proposition}[theorem]{Proposition}
\newtheorem{corollary}[theorem]{Corollary}

\theoremstyle{definition}

\newtheorem{assumption}[theorem]{Assumption}

\theoremstyle{remark}
\newtheorem{remark}[theorem]{Remark}

\newcommand{\bbB}{\mathbb{B}}

\newcommand{\bbR}{\mathbb{R}}

\newcommand{\bbZ}{\mathbb{Z}}

\newcommand{\ii}{\infty}
\newcommand\1{{\ensuremath {\mathds 1} }}
\renewcommand\phi{\varphi}

\newcommand{\cV}{\mathcal{V}}

\newcommand{\cW}{\mathcal{W}}

\newcommand{\cJ}{\mathcal{J}}
\newcommand{\cK}{\mathcal{K}}

\renewcommand{\geq}{\geqslant}
\renewcommand{\leq}{\leqslant}
\renewcommand{\hat}{\widehat}

\usepackage{color}

\newcommand{\dd}{\mathrm{d}}

\renewcommand{\d}[1]{\ensuremath{\operatorname{d}\!{#1}}}
\newcommand{\asym}{\textmd{\normalfont a}}
\newcommand{\sym}{\textmd{\normalfont s}}

\newcommand{\mdc}{{\textmd{c}}}

\author[A. Triay]{Arnaud Triay}
\address{Department of Mathematics, LMU Munich, Theresienstrasse 39, 80333 Munich, Germany}
\email{triay@math.lmu.de}
\author[F. L. A. Visconti]{François L. A. Visconti}
\address{Institute of Science and Technology Austria, Am Campus 1, 3400, Klosterneuburg, Austria}
\email{francois.visconti@ista.ac.at}

\title[Stability of Bose gases with repulsive three-body interactions]{Stability of attractive 3D Bose gases with repulsive three-body interactions}

\date{}

\begin{document}
 
\maketitle

\begin{abstract}
	We consider three-dimensional Bose gases with attractive two-body interactions and repulsive three-body interactions. We prove the stability of the second kind of the system for a large class of interaction potentials and scaling regimes, which includes the Gross--Pitaevskii regime for the three-body potential. In addition, we even prove the classical stability of the system for a more restrictive class of interactions and scaling regimes. The main ingredient of the proof is a bound that was derived recently in a study by the authors on the dynamics of Bose gases with three-body interactions in the Gross--Pitaevskii regime \cite{TriayVis26}.
\end{abstract}

\tableofcontents

\section{Introduction}

The stability of large quantum systems is an old topic that goes back to the works of Fisher and Ruelle \cite{Fisher64,FisherRue66,Ruelle63,Ruelle70}, where the distinction between classical and quantum stability was first made. On the one hand, an $N$-body interaction potential $U_N$ is said to be \textit{classically stable} if
\begin{equation*}
	U_N(x_1,\dots,x_N) \geq -CN,
\end{equation*}
for all $x_1,\dots,x_N\in\bbR^3$ and for some universal nonnegative constant $C$. On the other hand, the quantum equivalent to this is called \textit{stability of the second kind} and reads
\begin{equation*}
	\sum_{i = 1}^N-\Delta_i + U_N(x_1,\dots,x_N) \geq -CN
\end{equation*}
as a quadratic form on $L_{\sym/\asym}^2(\bbR^{3N})$---the symmetric/antisymmetric subspaces of $L^2(\bbR^{3N})$---again for some universal nonnegative constant $C$.

In the bosonic case, potentials that are stable of the second kind but not classically stable are notoriously rare. In particular, Fisher and Ruelle \cite{FisherRue66} proved that, under mild regularity assumptions, if the potential $U_N$ is of the form
\begin{equation*}
	U_N(x_1,\dots,x_N) = \sum_{1 \leq i,j\leq N}V(x_i - x_j),
\end{equation*}
for some fixed $V$, and if there exists an $N$ and a particle configuration $x_1,\dots,x_N\in\bbR^3$ such that
\begin{equation*}
	\sum_{1\leq i,j\leq N}V(x_i - x_j) < 0,
\end{equation*}
then the system is not stable of the second kind. A famous example is that of the two-component charged Bose gas whose ground state energy diverges as $-N^{7/5}$ \cite{ConlonLieYau88,Dyson67,LiebSol04}. Another example is that of bosons with gravitational forces whose ground state energy diverges as $-N^2$ \cite{LevyLebl69}. In contrast, fermions enjoy much better stability properties thanks to the Pauli exclusion principle. Notably, Dyson and Lenard \cite{DysonLen67,DysonLen68} proved that the two-component Fermi gas is stable of the second kind. Later, their approach was drastically simplified by Lieb and Thirring \cite{LiebThi75} who introduced important new mathematical inequalities---now known as Lieb–Thirring inequalities. Though these inequalities are in general false for bosons, Lundholm, Portmann and Solovej \cite{LundholmPorSol15}---based on \cite{LundholmSol14}---managed to prove Lieb--Thirring-type inequalities for Bose gases with repulsive pair interactions. One of the key elements in their proof is a local exclusion principle that results from the repulsion between the particles.

In the present work, we are interested in the stability of 3D Bose gases with (partially) attractive two-body interactions and repulsive three-body interactions. Such Bose gases have drawn considerable attention in the Physics literature \cite{Bla16,BisBla15,GamFreTomCho00,JosRic97,LahMetFroKoc08,Pet22} because repulsive three-body interactions can counterbalance the instability induced by attractive two-body interactions. Our main result is a rigorous justification of this fact under very general assumptions on the interaction, and for a wide range of scaling regimes, which includes the Gross--Pitaevskii regime for the three-body interaction potential. Furthermore, we also prove the classical stability of the system for a narrower range of scaling regimes. Notably, both results are true for very large attractive two-body potentials, as long as the three-body potential does not vanish in an arbitrary small region around the origin. This is in stark contrast to the situation without three-body interactions, where the system is unstable even for very small attractive potentials.

In the Mathematics literature, Bose gases with (partially) attractive two-body interactions have drawn considerable attention; we refer to \cite{JunVis26,LewNamRou14a,LewNamRou14,LewNamRou17,NamRou20,Tri18} for results on the ground state energy of such systems, and to \cite{BosDieNam23,Chong21,JebPic18,NamNap19,Pic10,Zha00} for results on their time evolution. Likewise, Bose gases with repulsive three-body interactions have been objects of numerous works; we refer to \cite{Bro25,DoanNgu26,JungeVis26,NamRicTri22a,NamRicTri22,NamRicTri23,Vis24} for results on the ground state energy, to \cite{BabRou26,ChePav11,Che11,CheHol19,LiYao21,NamSal19,Xie15,Yua15} for results on the dynamics, and to \cite{FarhatRou26,LueNamZhu26,RouSoh25} for results on the derivation of the Gibbs measure for the 1D quintic focusing NLS. Despite the extensive literature on the two different systems, our result is the first rigorous result on the stabilisation of attractive Bose gases by repulsive three-body interactions in 3D; see \cite{LiZha27,LiLiWan26,NguRic25,Nguyen26} for results in 2D for soft scaling regimes, and \cite{NguRic24} for a related result in 1D.
\\

\noindent
\textbf{Organisation of the paper.} In Section \ref{sec:main_result}, we describe the mathematical model and provide the main results: Theorems~\ref{th:stability_second_kind}~and~\ref{th:classical_stability}. These results are proved in Sections~\ref{sec:stability_second_kind_proof}~and~\ref{sec:classical_stability_proof}, respectively. The key ingredient of the proof of Theorem~\ref{th:stability_second_kind} is Proposition \ref{prop:general_two_body_potential_bound}, whose proof is given in Section~\ref{sec:proof_strategy_main_estimate}.

\section{Mathematical model and main results}

\label{sec:main_result}

We consider systems described by the Hamiltonian
\begin{equation*}
	H_N = \sum_{i = 1}^N-\Delta_i + \sum_{1 \leq i<j\leq N}V_N(x_i - x_j) + \sum_{1 \leq i<j<k\leq N}W_N(x_i - x_j, x_i - x_k)
\end{equation*}
acting on $L_\sym^2(\bbR^{3N})$, where the interaction potentials are of the form
\begin{equation}
	\label{eq:potentials_scaling}
	V_N = N^{3\alpha - 1}V(N^\alpha\cdot) \quad \textmd{and} \quad W_N = N^{6\beta - 2}W(N^\beta\cdot),
\end{equation}
for some fixed nonnegative parameters $\alpha$ and $\beta$, and some fixed functions $V$ and $W$ that satisfy the following properties.
\begin{assumption}
	\label{ass:stability_second_kind}
	The negative part $V_- = \max(0,-V)$ of $V$ is bounded from above by a radially symmetric and decreasing function with compact support and of class $L^{2+\varepsilon}(\bbR^3)$, for some $\varepsilon > 0$. Moreover, the potential $W$ is nonnegative and it has a positive soft-core, meaning that there exist positive numbers $R$ and $C$ such that
	\begin{equation}
		\label{eq:three_body_potential_bounded_below_assumption}
		W \geq C \quad \textmd{on $B(0,R)$.}
	\end{equation}
\end{assumption}
Note that the above conditions on $V_-$ in particular ensure that $H_N$ is bounded from below on the core domain $C_\mdc^\ii(\bbR^{3N})$. Our main result is the following.

\begin{theorem}[Stability of the second kind]
	\label{th:stability_second_kind}
	Let $0 \leq \alpha \leq \beta \leq 1/2$ and assume that $V$ and $W$ satisfy Assumptions~\ref{ass:stability_second_kind}. Then, there exists a nonnegative constant $C$ depending only on $V$ and $W$ such that
	\begin{equation}
		\label{eq:stability_second_kind}
		H_N \geq -CN
	\end{equation}
	holds on $L_\sym^2(\bbR^{3N})$.
\end{theorem}

Below are some comments on the main result. Since we are mainly interested in the purely attractive case $V \leq 0$, we only discuss that one:
\begin{enumerate}
	\item If $W = 0$, the Hamiltonian $H_N$ is not stable of the second kind as soon as $\alpha > 0$, regardless of how small $V$ is. Indeed, using a factorised trial state $u_N^{\otimes N}$ with $u_N = N^{3\alpha/2}u(N^\alpha\,\cdot\,)$, we find
	\begin{equation*}
		\dfrac{1}{N}\langle u_N^{\otimes N},H_Nu_N^{\otimes N}\rangle \leq N^{2\alpha}\int_{\bbR^3}\vert\nabla u\vert^2 + \dfrac{N^{3\alpha}}{2}\int_{\bbR^3}(V\ast\vert u\vert^2)\vert u\vert^2 \xrightarrow[N\rightarrow\ii]{} -\ii.
	\end{equation*}
	Hence, it is very nontrivial that \eqref{eq:stability_second_kind} holds for two-body potentials with an arbitrary large negative part as long as the three-body potential has a positive soft-core, which can be arbitrarily small.
	\item Though the assumption $\beta \geq \alpha$ might not be optimal, the Hamiltonian $H_N$ is not stable of the second kind if either $3\beta/2 < \alpha$ or $1 < \alpha$. To see why, consider a configuration consisting of clusters of $k$ particles localised in balls of size $N^{-\alpha}$, so that the total kinetic energy is of order $N^{1 + 2\alpha}$. Then, the number of clusters is proportional to $N/k$, which implies that the two-body interactions are of order $-kN^{3\alpha}$ and the three-body interactions of order $k^2N^{6\beta - 1}$. Optimising over $k$ (taking potentially $k = 2$ or $k = N$) leads to instability.
\end{enumerate}

The main difficulty is the regime $\beta > 1/3$, where the three-body interactions are too short-ranged to yield classical stability. Therefore, the kinetic operator is essential to ensure the stability of the system. More precisely, it combines with the repulsive three-body interactions to create a local exclusion principle, which counterbalances the collapse induced by the attractive two-body interactions. This statement is made rigorous in Proposition \ref{prop:general_two_body_potential_bound}, which was proved by the authors in \cite{TriayVis26}; we explain the proof strategy in Section~\ref{sec:proof_strategy_main_estimate}. This local exclusion principle bears similarities with the one used in \cite{LundholmPorSol15,LundholmSol14}.

When $\alpha \leq \beta \leq 1/3$, in addition to the stability of the second kind \eqref{eq:stability_second_kind}, we also establish the classical stability of the system under the following more restrictive assumptions on $V$.

\begin{assumption}
	\label{ass:classical_stability}
	The negative part $V_- = \max(0,-V)$ of $V$ is bounded from above and that there exists $C,\varepsilon > 0$ such that $V_-(x) \leq C\vert x\vert^{-3-\varepsilon}$. Moreover, the potential $W$ is nonnegative and it has a positive soft-core, meaning that there exist positive numbers $R$ and $C$ such that
	\begin{equation}
		\label{eq:three_body_potential_bounded_below_assumption_classical}
		W \geq C \quad \textmd{on $B(0,R)$.}
	\end{equation}
\end{assumption}

\begin{theorem}[Classical stability]
	\label{th:classical_stability}
	Let $0 \leq \alpha \leq \beta \leq 1/3$ and assume that $V$ and $W$ satisfy Assumption~\ref{ass:classical_stability}. Then, there exists a nonnegative constant $C$ depending only on $V$ and $W$ such that
	\begin{equation}
		\label{eq:classical_stability}
		\sum_{1 \leq i<j\leq N}V_N(x_i - x_j) + \sum_{1 \leq i<j<k\leq N}W_N(x_i - x_j, x_i - x_k) \geq -CN,
	\end{equation}
	for all $N\geq 1$ and $x_1,\dots,x_N\in\bbR^3$.
\end{theorem}

Here are comments on Theorem~\ref{th:classical_stability}, again in the purely attractive case $V \leq 0$:
\begin{enumerate}
	\item If $W = 0$, the system is not classically stable whenever $V(0) < 0$ for any $\alpha > 0$. This is easily seen by taking $x_1 = \dots = x_N$.
	\item The bound \eqref{eq:classical_stability} clearly implies \eqref{eq:stability_second_kind}, though the converse is in general not true, which is why stronger assumptions on the potential $V$ are required to prove classical stability.
	\item If either $\alpha > 1/3$ or $\beta > 1/3$, the system is not classically stable. In the former case, consider spread out pairs of particles at relative distances $N^{-\alpha}$ from each other. In the latter case with $\alpha \leq 1/3$, consider spread out particle clusters of size $N^{-\alpha}$ such that all particles are at a distance at least $N^{-\beta}$ from each other; there are approximately $N^{3(\beta - \alpha)}$ particles in each cluster.
\end{enumerate}

The main ingredient in the proof of Theorem \ref{th:classical_stability} is an estimate that shows that any two-body potential $V$ satisfying Assumption~\ref{ass:classical_stability} can be made classically stable by adding a sufficiently strong positive soft-core. This estimate is taken from \cite[Lemma 21]{Tri18} and stated in Proposition~\ref{prop:stabilisation_soft_core}; we give the proof in Appendix~\ref{app:stabilisation_soft_core}. In turn, this allows us to bound the potential $V_N$ by a more singular two-body potential $N^{3\beta - 1}\tilde{V}(N^\beta\cdot)$, which can be controlled by the three-body potential $W_N$.
\\

\noindent
\textbf{Acknowledgments.} This work was partly funded by the Deutsche Forschungsgemeinschaft (DFG, German Research Foundation) through the TRR 352 Project ID	47090307. F. L. A. V. acknowledges partial support from the European Research Council through the ERC CoG RAMBAS Project Nr. 101044249.
\\

\noindent
\textbf{AI usage statement.} AI tools were used solely for proofreading purposes, and in Sections~\ref{sec:stability_second_kind_less_singular}~and~\ref{sec:stability_second_kind_least_singular} to optimise the choice of the different parameters. All proofs were written and checked by the authors, who take full responsibility for the results.

\section{Proof of Theorem \ref{th:stability_second_kind}}

\label{sec:stability_second_kind_proof}

Throughout this section, we assume without loss of generality that $V$ is nonpositive, radially symmetric nondecreasing, and with compact support. The proof of Theorem \ref{th:stability_second_kind} is split between the regimes $\beta \geq \alpha \geq 1/3$, $\beta \geq 1/3 \geq \alpha$ and $1/3 \geq \beta \geq \alpha$, which are covered in Sections~\ref{sec:stability_second_kind_most_singular},~\ref{sec:stability_second_kind_less_singular}~and~\ref{sec:stability_second_kind_least_singular}, respectively. Before that, we provide some preliminary results in Section~\ref{sec:stability_second_kind_preliminaries}.

\subsection{Local exclusion bounds and layer cake representation}

\label{sec:stability_second_kind_preliminaries}

The main technical ingredients in the proof of Theorem~\ref{th:stability_second_kind} are local exclusion bounds, which are stated in the following proposition. The estimate \eqref{eq:general_two_body_potential_bound} is taken from \cite[Proposition 4.1]{TriayVis26}, and the bounds \eqref{eq:general_two_body_potential_bound_weaker} and \eqref{eq:two_body_potential_more_singuilar} are proved in a similar way; see Section \ref{sec:proof_strategy_main_estimate} for the proof strategy of \ref{eq:general_two_body_potential_bound} as well as the proofs of \eqref{eq:general_two_body_potential_bound_weaker} and \eqref{eq:two_body_potential_more_singuilar}.

\begin{proposition}[Local exclusion bounds]
	\label{prop:general_two_body_potential_bound}
	Let $\ell_1,\ell_2,\ell_3 > 0$ satisfy $\ell_1 \geq \ell_2$ and $\ell_1 \geq \ell_3$. Define
	\begin{equation*}
		\cK = \sum_{i = 1}^N-\Delta_i, \quad \cV^{(\ell_2)} = \sum_{1\leq i< j\leq N}\mathds{1}(|x_i - x_j|\leq \ell_2), \quad \cV^{(\ell_3)} = \sum_{1\leq i< j\leq N}\mathds{1}(|x_i - x_j|\leq \ell_3)
	\end{equation*}
	and
	\begin{equation*}
		\cW^{(\ell_3)} = \sum_{1\leq i<j<k \leq N}\mathds{1}(|(x_i - x_j,x_i - x_k)|\leq\ell_3).
	\end{equation*}
	Then, there exists a universal constant $C$ such that the inequalities
	\begin{equation}
		\label{eq:general_two_body_potential_bound}
		\cV^{(\ell_2)} \leq C\Big(\frac{\ell_2}{\ell_1}\Big)^3N + C\ell_2^2\Big(1 + \dfrac{\ell_2}{\ell_3}\Big)\Big(\frac{\ell_1}{\ell_3}\Big)^3\cK + C\Big(1 + \Big(\frac{\ell_2}{\ell_3}\Big)^3\Big)\Big(\frac{\ell_1}{\ell_3}\Big)^3\cW^{(\ell_3)},
	\end{equation}
	\begin{equation}
		\label{eq:general_two_body_potential_bound_weaker}
		\cV^{(\ell_2)} \leq C\eta\Big(\dfrac{\ell_2}{\ell_3}\Big)^3N + C\eta\ell_2^2\Big(\dfrac{\ell_1}{\ell_3}\Big)^3\cK + C\eta^{-1}\Big(\dfrac{\ell_1}{\ell_3}\Big)^3\cW^{(\ell_3)},
	\end{equation}
	for any $\eta \geq C$, and
	\begin{equation}
		\label{eq:two_body_potential_more_singuilar}
		\cV^{(\ell_2)} \leq C\Big(\frac{\ell_1}{\ell_3}\Big)^3N + C\Big(\frac{\ell_1}{\ell_3}\Big)^3\cV^{(\ell_3)}
	\end{equation}
	hold on $L_\sym^2(\bbR^{3N})$.
\end{proposition}

\begin{remark}
	When $V$ is bounded from below, Theorem~\ref{th:stability_second_kind} follows almost immediately from Proposition~\ref{prop:general_two_body_potential_bound}. When $\beta \geq \alpha \geq 1/3$, we use \eqref{eq:general_two_body_potential_bound} with
	\begin{equation*}
		\ell_1 \simeq N^{-1/3}, \quad \ell_2 \simeq N^{-\alpha} \quad \textmd{and} \quad \ell_3 \simeq N^{-\beta}.
	\end{equation*}
	When $\beta \geq 1/3 \geq \alpha$, we use \eqref{eq:two_body_potential_more_singuilar} with
	\begin{equation*}
		\ell_1 = \ell_2 \simeq N^{-\alpha} \quad \textmd{and} \quad \ell_3 \simeq N^{-1/3}
	\end{equation*}
	in order to bound $V_N$ by a new two-body potential with range of order $N^{-1/3}$, and then we apply the result from the previous case to that potential. When $1/3 \geq \beta \geq \alpha$, we use \eqref{eq:general_two_body_potential_bound_weaker} with
	\begin{equation*}
		\ell_1 = \ell_2 \simeq N^{-\alpha}, \quad \ell_3 \simeq N^{-\beta}, \quad \textmd{and} \quad \eta = N^{1 - 3\beta}.
	\end{equation*}
\end{remark}

Now, we collect some preliminary results that will be useful in the proof of Theorem~\ref{th:stability_second_kind}. The layer cake representation (see e.g. \cite[Theorem 1.13]{LieLos01}) of $V$ reads
\begin{equation*}
	V(x) = -\int_0^\ii\dd{t}\1(V(x) < -t).
\end{equation*}
Recall that $V$ is radially symmetric and decreasing with compact support. Thus, we can find a nonincreasing bounded function $t \mapsto R_t$ such that
\begin{equation}
	\label{eq:layer_cake}
	V(x) = -\int_0^\ii\dd{t}\1(\vert x\vert \leq R_t)
\end{equation}
(see Figure \ref{fig:layer_cake}). Moreover, the $L^p$ norm of $V$ can be expressed as
\begin{equation}
	\label{eq:layer_cake_Lp_norm}
	\Vert V\Vert_p^p = p\int_0^\ii\dd{t}t^{p-1}\left\vert\{x\in\bbR^3:V(x) < -t\}\right\vert = p\vert\bbB^3\vert\int_0^\ii\dd{t}t^{p-1}R_t^3,
\end{equation}
for any $p \geq 1$. Here, $\vert\bbB^3\vert$ denotes the volume of the unit ball in three dimensions.

\begin{figure}[ht]
	\centering
	\begin{tikzpicture}
	\begin{axis}
		[
		axis lines=middle,
		axis line style={thick},
		y label style={anchor=east},
		ylabel={\color{Blue}$V$},
		xmin=-1.2, xmax=1.2,
		ymin=-3, ymax=0.5,
		xtick=\empty,
		ytick=\empty,
		samples=400,
		domain=-1:1,
		width=10cm,
		height=6cm,
		]
		
		\draw[Red,very thick] (0.35,{-exp(-1/(1-0.35^2))/abs(0.35)}) -- (-0.35,{-exp(-1/(1-0.35^2))/abs(0.35)});
		
		\filldraw[Red] (0,{-exp(-1/(1-0.35^2))/abs(0.35)}) circle (2pt) node[Red,anchor=north east] {\large $-t$};
		
		\draw[Red,very thick,dashed] (0.35,{-exp(-1/(1-0.35^2))/abs(0.35)}) -- (0.35,0);
		\draw[Red,very thick,dashed] (-0.35,{-exp(-1/(1-0.35^2))/abs(0.35)}) -- (-0.35,0);
		
		\node[Red,anchor=south] at (0.35,0) {\large $R_t$};
		
		\draw[Red, very thick](-0.35,-0.05) -- (-0.35,0.05);
		\draw[Red, very thick](0.35,-0.05) -- (0.35,0.05);
		
		\draw[Blue, very thick](-1,-0.075) -- (-1,0.075);
		\draw[Blue, very thick](1,-0.075) -- (1,0.075);
		
		\addplot[very thick, Blue, domain=-0.999:-0.1, samples=200]
		{-exp(-1/(1-x^2))/abs(x)};
		\addplot[very thick, Blue, domain=0.1:0.999, samples=200]
		{-exp(-1/(1-x^2))/abs(x)};
		
	\end{axis}
\end{tikzpicture}
	\caption{Illustration of the layer cake representation of some nonpositive, radially symmetric nondecreasing, compactly supported potential $V$.}
	\label{fig:layer_cake}
\end{figure}
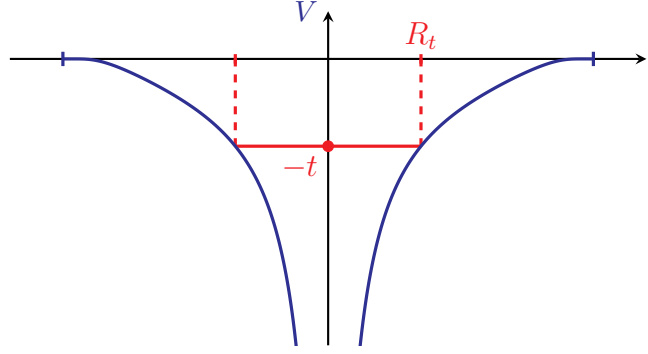

Since we will apply Proposition~\ref{prop:general_two_body_potential_bound} to the integrand in the right-hand side of \eqref{eq:layer_cake}, it is convenient to rewrite the estimates \eqref{eq:general_two_body_potential_bound}, \eqref{eq:general_two_body_potential_bound_weaker} and \eqref{eq:two_body_potential_more_singuilar} in terms of this indicator function. To this end, we define
\begin{equation*}
	\cK = \sum_{i = 1}^N-\Delta_i, \quad \cV_{N,t} = \sum_{1 \leq i<j\leq N}N^{3\alpha - 1}\1(\vert x_i - x_j\vert \leq R_tN^{-\alpha})
\end{equation*}
\begin{equation*}
	\tilde{\cV}_N = \sum_{1 \leq i<j\leq N}\1(\vert x_i - x_j\vert \leq \tilde{R}N^{-1/3}) \quad \textmd{and} \quad \cW_N = \sum_{1 \leq i<j<k\leq N}W_N(x_i - x_j,x_i - x_k).
\end{equation*}
Here, the two-body potential $\tilde{\cV}_N$ is defined for some $\tilde{R} > 0$ to be specified. Taking $\ell_2 = R_tN^{-\alpha}$ and $\ell_3 = RN^{-\beta}$, the bound \eqref{eq:general_two_body_potential_bound} becomes
\begin{equation}
	\label{eq:two_body_potential_bound_layer_cake}
	\cV_{N,t} \leq CR_t^3\ell_1^{-3} + CR_t^2\ell_1^3N^{4\beta - 1}\cK + C\ell_1^3N\cW_N,
\end{equation}
for any $\ell_1$ that satisfies
\begin{equation*}
	\ell_1 \geq R_tN^{-\alpha} \quad \textmd{and} \quad \ell_1 \geq RN^{-\beta}.
\end{equation*}
Likewise, the estimate \eqref{eq:general_two_body_potential_bound_weaker} reads
\begin{equation}
	\label{eq:two_body_potential_bound_layer_cake_weaker}
	\cV_{N,t} \leq CR_t^3\eta N^{3\beta} + CR_t^2\eta \ell_1^3N^{3\beta + \alpha - 1}\cK + C\eta^{-1}\ell_1^3N^{3\alpha -3\beta + 1}\cW_N,
\end{equation}
under the conditions
\begin{equation*}
	\ell_1 \geq R_tN^{-\alpha}, \quad \ell_1 \geq RN^{-\beta} \quad \textmd{and} \quad \eta \geq C.
\end{equation*}
Lastly, the estimate \eqref{eq:two_body_potential_more_singuilar} implies
\begin{equation}
	\label{eq:two_body_potential_singularisation_layer_cake}
	\cV_{N,t} \leq C\tilde{R}^{-3}\ell_1^3N^{3\alpha + 3\delta} + C\tilde{R}^{-3}\ell_1^3N^{3\alpha}\tilde{\cV}_N,
\end{equation}
for any $\ell_1$ satisfying
\begin{equation*}
	\ell_1 \geq R_tN^{-\alpha} \quad \textmd{and} \quad \ell_1 \geq \tilde{R}N^{-\delta}.
\end{equation*}

\subsection{Proof of Theorem \ref{th:stability_second_kind}: case $\beta \geq \alpha \geq 1/3$}

\label{sec:stability_second_kind_most_singular}

Let $\varepsilon > 0$ be fixed. We show that
\begin{equation}
	\label{eq:stability_second_kind_epsilon}
	\cV_N \leq C_\varepsilon N + \varepsilon\cK + \varepsilon\cW_N.
\end{equation}
To do so, we treat the following three cases separately:
\begin{enumerate}[label=(\alph*)]
	\item \label{eq:most_singular_case_1} $\varepsilon^{1/3}R_t^{1/2}N^{-1/3} \geq RN^{-\beta}$;
	\item \label{eq:most_singular_case_2} $RN^{-\beta} \geq R_tN^{-\alpha} \geq \varepsilon^{-2}N^{-3\beta}$;
	\item \label{eq:most_singular_case_3} $\varepsilon^{-2}N^{-3\beta} \geq R_tN^{-\alpha}$.
\end{enumerate}

On the one hand, in the cases \ref{eq:most_singular_case_1} and \ref{eq:most_singular_case_2}, we claim that
\begin{equation}
	\label{eq:general_two_body_potential_bound_singular_case}
	\cV_{N,t} \leq CR_t^{3/2}(C_\varepsilon N + \varepsilon \cK + \varepsilon \cW_N).
\end{equation}
Indeed, this bound follows from an application of \ref{eq:two_body_potential_bound_layer_cake} with
\begin{equation*}
	\ell_1 = \varepsilon^{1/3}R_t^{1/2}N^{-1/3}
\end{equation*}
when \ref{eq:most_singular_case_1} holds, and it instead follows from an application of \eqref{eq:two_body_potential_bound_layer_cake_weaker} with
\begin{equation*}
	\ell_1 = RN^{-\beta} \quad \textmd{and} \quad \eta = \dfrac{R^3}{\varepsilon R_t^{3/2}}N^{-3\beta + 1}
\end{equation*}
when \ref{eq:most_singular_case_2} holds. On the other hand, under the condition \ref{eq:most_singular_case_3}, the Sobolev inequality gives
\begin{equation*}
	\cV_{N,t} \leq \varepsilon CR_t^{3/2-\tau}\cK,
\end{equation*}
for any $\tau > 0$ small enough (depending on $\varepsilon$).

In summary, we have shown that, for any $\varepsilon > 0$ and $\tau >0$ small enough, the bound
\begin{equation}
	\label{eq:two_body_potential_bound_most_singular}
	\cV_{N,t} \leq CR_t^{3/2-\tau}(C_\varepsilon N + \varepsilon \cK + \varepsilon \cW_N),
\end{equation}
holds for any $t > 0$. Finally, the bound \eqref{eq:stability_second_kind_epsilon} follows from an integration over $t\in(0,\infty)$, as well as the identity \eqref{eq:layer_cake} and the fact that, for $\tau$ small enough, the integral of $R_t^{3/2 - \tau}$ is a constant. In fact, an application of Hölder's inequality gives
\begin{equation*}
	\int_0^\ii\d{t}R_t^{3/2 - \tau} \leq \Big(\int_0^\ii\d{t}(1 + t)^{1 + \eta}R_t^3\Big)^{\frac{1}{2} - \frac{\tau}{3}}\Big(\int_0^\ii\d{t}(1 + t)^{-\left(1 + \eta)(1 - \frac{4\tau}{3 + \tau}\right)}\Big)^{\frac{1}{2} + \frac{\tau}{3}},
\end{equation*}
which, in view of \eqref{eq:layer_cake_Lp_norm}, is a finite constant that depends only on $V$ for any fixed $\tau$ small enough. This proves \eqref{eq:stability_second_kind_epsilon}, which finishes the proof of Theorem \ref{th:stability_second_kind} when $\beta \geq \alpha \geq 1/3$.

\subsection{Proof of Theorem \ref{th:stability_second_kind}: case $\beta \geq 1/3 \geq \alpha$}

\label{sec:stability_second_kind_less_singular}

Let $\varepsilon > 0$ be fixed. Again, we show
\begin{equation}
	\label{eq:stability_second_kind_epsilon_less_singular}
	\cV_N \leq C_\varepsilon N + \varepsilon\cK + \varepsilon\cW_N.
\end{equation}
For this, we make a distinction between:
\begin{enumerate}[label=(\alph*)]
	\item \label{eq:less_singular_case_1} $R_tN^{-\alpha} \geq \varepsilon^{1/3}R_t^{1/2}N^{-1/3}$;
	\item \label{eq:less_singular_case_2} $\varepsilon^{1/3}R_t^{1/2}N^{-1/3} \geq R_tN^{-\alpha}$.
\end{enumerate}
To deal with the case \ref{eq:less_singular_case_2}, we proceed in exactly the same way as in Section~\ref{sec:stability_second_kind_most_singular}. This shows that
\begin{equation}
	\label{eq:two_body_potential_bound_less_singular}
	\cV_{N,t} \leq CR_t^{3/2-\tau}(C_\varepsilon N + \varepsilon \cK + \varepsilon \cW_N),
\end{equation}
holds here as well. In the case \ref{eq:less_singular_case_1}, we use \eqref{eq:two_body_potential_singularisation_layer_cake} with $\tilde{R} = \varepsilon R_t^{1/2}$ to deduce
\begin{equation*}
	\cV_{N,t} \leq C_\varepsilon R_t^{3/2}(N + \tilde{\cV}_N),
\end{equation*}
where $\tilde{\cV}_N$ is given by
\begin{equation*}
	\tilde{\cV}_N = \sum_{1 \leq i<j\leq N}\1(\vert x_i - x_j\vert \leq \varepsilon \sup_tR_t^{1/2}N^{-1/3}).
\end{equation*}
The assumption that $V$ has compact support implies that $\sup_tR_t^{1/2}$ is finite, which means that $\tilde{\cV}_N$ is a two-body potential with range of order $N^{-1/3}$. Therefore, we may apply \eqref{eq:stability_second_kind_epsilon} to $\tilde{\cV}_N$in order to deduce that \eqref{eq:two_body_potential_bound_less_singular} also holds under the condition \ref{eq:less_singular_case_1}.

Having established \eqref{eq:two_body_potential_bound_less_singular} for all $t > 0$, the rest of the proof is identical to that of Section~\ref{sec:stability_second_kind_most_singular}, which concludes the proof of Theorem~\ref{th:stability_second_kind} for $\beta \geq 1/3 \geq \alpha$.

\subsection{Proof of Theorem \ref{th:stability_second_kind}: case $1/3 \geq \beta \geq \alpha$}

\label{sec:stability_second_kind_least_singular}

Let $\varepsilon > 0$ be fixed. Just like in the previous two sections, we show
\begin{equation}
	\label{eq:stability_second_kind_epsilon_least_singular}
	\cV_N \leq C_\varepsilon N + \varepsilon\cK + \varepsilon\cW_N.
\end{equation}
Without loss of generality, we assume that $\alpha > 0$. In fact, for $\alpha = 0$, the bound \eqref{eq:stability_second_kind_epsilon_least_singular} can be shown to hold with $\cW_N = 0$ by splitting between $V \geq -\Lambda$ and $V \leq -\Lambda$, for some parameter $\Lambda > 0$, and using the Sobolev inequality in the second situation. In order to prove \eqref{eq:stability_second_kind_epsilon_least_singular}, we make the following distinction:
\begin{enumerate}[label=(\alph*)]
	\item \label{eq:least_singular_case_1} $R_tN^{-\alpha} \geq RN^{-\beta}$
	\item \label{eq:least_singular_case_2} $RN^{-\beta} \geq R_tN^{-\alpha} \geq \varepsilon^{-2}N^{-3\beta}$
	\item \label{eq:least_singular_case_3} $\varepsilon^{-2}N^{-3\beta} \geq R_tN^{-\alpha}$
\end{enumerate}
In the case \ref{eq:least_singular_case_1}, we apply \ref{eq:two_body_potential_bound_layer_cake_weaker} with
\begin{equation*}
	\ell_1 = R_tN^{-\alpha} \quad \textmd{and} \quad \eta = \dfrac{1}{\varepsilon}N^{1-3\beta}
\end{equation*}
to deduce
\begin{equation*}
	\cV_{N,t} \leq CR_t^3(C_\varepsilon N + \varepsilon^{-1}R_t^2N^{-2\alpha}\cK + \varepsilon\cW_N) \leq CR_t^3(C_\varepsilon N + \varepsilon \cK + \varepsilon\cW_N),
\end{equation*}
for $N$ large enough. In the case \ref{eq:least_singular_case_2}, we use \eqref{eq:two_body_potential_bound_layer_cake_weaker} with
\begin{equation*}
	\ell_1 = RN^{-\beta} \quad \textmd{and} \quad \eta = \dfrac{1}{\varepsilon}\dfrac{R^{3/2}}{R_t^{3/2}}N^{1 - 3\beta}
\end{equation*}
to get
\begin{equation*}
	\cV_{N,t} \leq CR_t^{3/2}(C_\varepsilon N + \varepsilon\cK + \varepsilon\cW_N).
\end{equation*}
Finally, the case \ref{eq:least_singular_case_3} is the same as the case \ref{eq:most_singular_case_3} from Section~\ref{sec:stability_second_kind_most_singular}, which is dealt with using the Sobolev inequality. This proves that
\begin{equation}
	\label{eq:two_body_potential_bound_least_singular}
	\cV_{N,t} \leq CR_t^{3/2-\tau}(C_\varepsilon N + \varepsilon \cK + \varepsilon \cW_N),
\end{equation}
is true for any $t > 0$. Therefore, we may conclude the proof just as in Section~\ref{sec:stability_second_kind_most_singular}, which proves Theorem~\ref{th:stability_second_kind} for $1/3\geq \beta \geq \alpha$.

\section{Proof of Theorem \ref{th:classical_stability}}

\label{sec:classical_stability_proof}

The proof of Theorem~\ref{th:classical_stability} occurs in two steps. First, we show that $V_N$ is bounded from below by a two-body potential with range of order $N^{-\beta}$. Then, we use an elementary estimate to show that this more singular two-body potential is controlled by the three-body potential $W_N$. This is done in Sections~\ref{sec:classical_stability_singularisation}~and~\ref{sec:stability_second_kind_proof_subsec}, respectively.

\subsection{Stabilisation by adding a soft-core and singularisation}

\label{sec:classical_stability_singularisation}

The main technical ingredient in the proof of Theorem~\ref{th:classical_stability} is the following result, which shows that the potential $V$ can be made classically stable by adding a sufficiently strong soft-core. This was proved by one of the authors in \cite[Lemma 21]{Tri18}, borrowing ideas from \cite[Theorem IV]{FisherRue66}; for the reader's convenience, we recall the proof in Appendix~\ref{app:stabilisation_soft_core}.

\begin{proposition}[Stabilisation by adding a soft-core]
	\label{prop:stabilisation_soft_core}
	Let $V$ satisfy Assumption~\ref{ass:classical_stability}, and let $R_0 > 0$. Then, there exists a nonnegative constant $C$ such that
	\begin{equation*}
		\sum_{1 \leq i < j\leq N}V(x_i - x_j) \geq -C\sum_{1 \leq i<j\leq N}\1(\vert x_i - x_j\vert \leq R_0) - CN,
	\end{equation*}
	for all $N \geq 1$ and $x_1,\dots,x_N\in\bbR^3$.
\end{proposition}

As an important consequence of Proposition~\ref{prop:stabilisation_soft_core}, we have the following corollary. Notably, this allows us to bound the two-body potential $V_N$ from below by a more singular compactly supported two-body potential $N^{3\beta - 1}\tilde{V}(N^\beta\cdot)$, which is particularly useful to compare with the three-body potential $W_N$.

\begin{corollary}[Singularisation]
	\label{cor:singularisation}
	Let $V$ satisfy Assumption~\ref{ass:classical_stability}, and, given $0 \leq \alpha \leq 1/3$, define $V_N$ as in \eqref{eq:potentials_scaling}. Let $\delta$ satisfy $\alpha \leq \delta \leq 1/3$ and let $R_0 > 0$. Then, there exists a nonnegative constant $C$ such that
	\begin{equation}
		\label{eq:soft_two_body_potential_singularisation}
		\sum_{1\leq i<j\leq N}V_N(x_i - x_j) \geq -\sum_{1 \leq i<j\leq N}CN^{3\delta - 1}\1(N^\delta\vert x_i - x_j\vert \leq R_0) - CN.
	\end{equation}
\end{corollary}

\begin{proof}[Proof of Corollary~\ref{cor:singularisation}]
	Define $\tilde{V}:x\in\bbR^3\mapsto e^{-\vert x\vert^2/2}$ and denote
	\begin{equation*}
		\tilde{V}_N^{(\alpha)} = N^{3\alpha - 1}\tilde{V}(N^{\alpha}\,\cdot\,) \quad \textmd{and} \quad \tilde{V}_N^{(\delta)} = N^{3\delta - 1}\tilde{V}(N^{\delta}\,\cdot\,).
	\end{equation*}
	Notice first that, thanks to Proposition~\ref{prop:stabilisation_soft_core}, there exists a nonnegative constant $C$ such that $V + C\tilde{V}$ is classically stable, which implies
	\begin{equation*}
		\sum_{1 \leq i<j\leq N}V_N(x_i - x_j) \geq -\sum_{1 \leq i<j\leq N}C\tilde{V}_N^{(\alpha)}(x_i - x_j) - CN^{3\alpha},
	\end{equation*}
	for all $x_1,\dots,x_N\in\bbR^3$. Next, we use that the Fourier transform of $\tilde{V}$ is positive, radially symmetric and decreasing (see e.g. \cite[Theorem 5.2]{LieLos01}) to deduce that $\tilde{V}_N^{(\delta)} - V_N^{(\alpha)}$ has a positive Fourier transform. However, this implies that this potential is classically stable, which yields
	\begin{equation*}
		\sum_{1 \leq i<j\leq N}V_N(x_i - x_j) \geq -\sum_{1 \leq i<j\leq N}C\tilde{V}_N^{(\delta)}(x_i - x_j) - CN,
	\end{equation*}
	for all $x_1,\dots,x_N\in\bbR^3$. Finally, we apply Proposition~\ref{prop:stabilisation_soft_core} to $-\tilde{V}$ to deduce \eqref{eq:soft_two_body_potential_singularisation}.
\end{proof}

\subsection{Proof of Theorem~\ref{th:classical_stability}}

\label{sec:stability_second_kind_proof_subsec}

Let $R_0 > 0$ to be chosen later. An application of \eqref{eq:soft_two_body_potential_singularisation} with $\delta = \beta$ gives
\begin{equation*}
	\sum_{1 \leq i<j\leq N}V_N(x_i - x_j) \geq - \sum_{1\leq i<j\leq N}CN^{3\beta - 1}\1(\vert x_i - x_j\vert \leq R_0N^{-\beta}) - CN.
\end{equation*}
Then, we use the Cauchy--Schwarz inequality to find, for all $\eta > 0$,
\begin{multline*}
	\sum_{\substack{1\leq i,j\leq N\\ i\neq j}}N^{3\beta - 1}\1(\vert x_i - x_j\vert \leq R_0N^{-\beta}) \leq \eta^{-1}N + \eta\sum_{i = 1}^N\bigg(\sum_{\substack{j = 1\\ j\neq 1}}^NN^{3\beta - 1}\1(\vert x_i - x_j\vert \leq R_0N^{-\beta})\bigg)^2\\
	\begin{aligned}[t]
		&= \eta^{-1}N + \eta\sum_{\substack{1\leq i,j,k\leq N\\ i\neq i\neq j\neq k\neq i}}N^{6\beta - 2}\1(\vert x_i - x_j\vert \leq R_0N^{-\beta})\1(\vert x_i - x_k\vert \leq R_0N^{-\beta})\\
		&\phantom{=} + \eta\sum_{\substack{1 \leq i,j\leq N\\ i\neq j}}N^{6\beta - 2}\1(\vert x_i - x_j\vert \leq R_0N^{-\beta}).
	\end{aligned}
\end{multline*}
Since $\beta \leq 1/3$, we may move the last term to the left-hand side, which yields a prefactor $1 - \eta N^{3\beta - 1}$. Taking $R_0$ small enough and using \eqref{eq:three_body_potential_bounded_below_assumption}, we can bound the second term from above by $W_N$. Thus, the estimate \eqref{eq:classical_stability} follows from an appropriate choice of $\eta$. This finishes the proof of Theorem~\ref{th:classical_stability}.

\section{Local exclusion principle: proof of Proposition~\ref{prop:general_two_body_potential_bound}}

\label{sec:proof_strategy_main_estimate}

This section is entirely devoted to the proof of Proposition~\ref{prop:general_two_body_potential_bound}. Since the estimate \eqref{eq:general_two_body_potential_bound} was already proved by the authors in \cite[Proposition 4.1]{TriayVis26}, we only need to prove \eqref{eq:general_two_body_potential_bound_weaker} and \eqref{eq:two_body_potential_more_singuilar}, which we do in that order. For the reader's convenience, we also explain how the proof of \eqref{eq:general_two_body_potential_bound_weaker} must be modified in order to prove \eqref{eq:general_two_body_potential_bound}.

\begin{proof}[Proof of \eqref{eq:general_two_body_potential_bound_weaker} from Proposition~\ref{prop:general_two_body_potential_bound}]
	The core idea is to localise the system into boxes of a well-chosen size and make a distinction according to the number of particles contained in each of these boxes. On the one hand, in boxes that contain many particles, we use \eqref{eq:three_body_potential_bounded_below_assumption} to argue that the number of three-body collisions scales cubically with the number of particles, and thus controls the number of two-body collisions, which scales quadratically. On the other hand, in boxes that contain few particles, we use the Poincaré--Sobolev inequality to show that the two-body interactions are controlled by the kinetic energy and the particle number $N$.
	\\
	
	\begin{figure}[ht]
		\centering
		\begin{tikzpicture}
	\tikzmath{\w = 6; \se = \w/3;}
	
	\draw[gray, thick] (-\w/2,-\w/2 + \se) -- (\w/2,-\w/2 + \se);
	\draw[gray, thick] (-\w/2,-\w/2 + 2*\se) -- (\w/2,-\w/2 + 2*\se);
	\draw[gray, thick] (-\w/2 + \se,-\w/2) -- (-\w/2 + \se,\w/2);
	\draw[gray, thick] (-\w/2 + 2*\se,-\w/2) -- (-\w/2 + 2*\se,\w/2);
	
	\draw[Red, ultra thick] (-\w/2,-\w/2) -- (-\w/2,\w/2) -- (\w/2,\w/2) -- (\w/2,-\w/2) -- cycle;
	
	\draw[Blue, ultra thick] (-\w/2 + \se,-\w/2 + \se) -- (-\w/2 + \se,-\w/2 + 2*\se) -- (-\w/2 + 2*\se,-\w/2 + 2*\se) -- (-\w/2 + 2*\se,-\w/2 + \se) -- cycle;
	
	\filldraw[black] (0,0) circle (2pt) node[anchor=north east]{\large $r$};
	\draw[color=black, very thick](-\w/8,\w/8) circle (\w/4);
	\node[color=black](a) at ({-\w/8 - \w/(2*4*sqrt(2))},{\w/8- \w/(2*4*sqrt(2))}) {\large $\ell_2$};
	\draw[color=black, very thick,dotted](-\w/8,\w/8) -- (a);
	\draw[color=black, very thick,dotted](a) -- ({-\w/8 - \w/(4*sqrt(2))},{\w/8- \w/(4*sqrt(2))});
	
	\filldraw[black] (-\w/8,\w/8) circle (2pt) node[anchor=north west]{\large $x_i$};
	\filldraw[black] (-\w/4-0.1,\w/4) circle (2pt) node[anchor=north west]{\large $x_j$};
	
	\node[color=Red](a) at (-\w/2 - 0.5,0) {\large $3\ell_1$};
	\draw[color=Red, ultra thick, dashed,<-](-\w/2 - 0.5,\w/2) -- (a);
	\draw[color=Red, ultra thick, dashed,->](a) -- (-\w/2 - 0.5,-\w/2);
	\node[color=Blue](b) at (\w/2 + 0.5,0) {\large $\ell_1$};
	\draw[color=Blue, ultra thick, dashed,<-](\w/2 + 0.5,\se/2) -- (b);
	\draw[color=Blue, ultra thick, dashed,->](b) -- (\w/2 + 0.5,-\se/2);
\end{tikzpicture}
		\caption{Illustration of the boxes ${\color{Blue}\Lambda_r^{(\ell_1)}}$ and ${\color{Red}\Lambda_r^{(3\ell_1)}}$. Because the interaction has range $\ell_2$, which is less than $\ell_1$, particles located in the box ${\color{Blue}\Lambda_r^{(\ell_1)}}$ can only interact with particles located in the box ${\color{Red}\Lambda_r^{(3\ell_1)}}$.}
		\label{fig:box_decomposition}
	\end{figure}
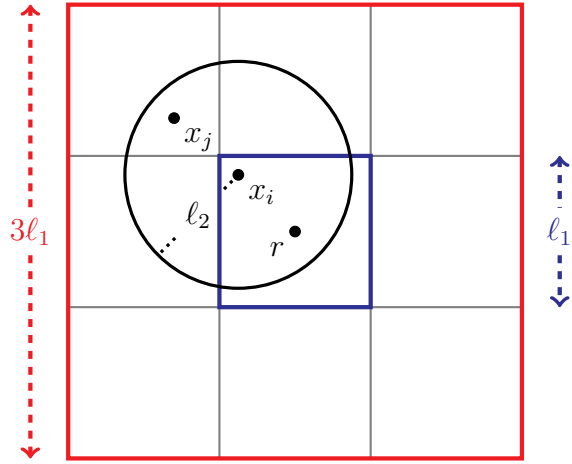
	
	\noindent
	\textbf{Step 1: Localisation into boxes.} Let $\Lambda_1 = [-1/2,1/2)^3$ and define
	\begin{equation}
		\label{eq:box_def}
		\Lambda_r^{(\ell)} = r + \ell\Lambda_1,
	\end{equation}
	the box of centre $r$ and of side length $\ell$. To bound the left-hand side of \eqref{eq:general_two_body_potential_bound_weaker}, we begin by decomposing the space $\bbR^3$ into boxes of side length $\ell_1$. Namely, we decompose
	\begin{equation*}
		\1(\vert x_i - x_j\vert\leq \ell_2) = \sum_{r\in\ell_1\bbZ^3}\1(\vert x_i - x_j\vert \leq \ell_2)\1_{\Lambda_r^{(\ell_1)}}(x_i).
	\end{equation*}
	Using that, if $x_i$ is in $\Lambda_r^{(\ell_1)}$, then $x_j$ must belong to $\Lambda_r^{(3\ell_1)}$ in order to be at a distance less than $\ell_2$ from $x_i$, we can bound
	\begin{equation*}
		\1(\vert x_i - x_j\vert\leq \ell_2) \leq \sum_{r\in\ell_1\bbZ^3}\1(\vert x_i - x_j\vert \leq \ell_2)\1_{\Lambda_r^{(3\ell_1)}}(x_i)\1_{\Lambda_r^{(3\ell_1)}}(x_j)
	\end{equation*}
	(see Figure \ref{fig:box_decomposition}). Note that we also used the trivial inclusion $\Lambda_r^{(\ell_1)} \subset \Lambda_r^{(3\ell_1)}$. After that, we make a distinction according to the number of particles contained in each box. To this end, define
	\begin{equation*}
		N_r^{(\ell)} = \sum_{i = 1}^N\1_{\Lambda_r^{(\ell)}}(x_i),
	\end{equation*}
	the number of particles contained in the box $\Lambda_r^{(\ell)}$. Though this depends on $(x_1,\dots,x_N)\in\bbR^{3N}$, we omit this from the notation for the sake of readability. Define also
	\begin{equation*}
		\eta_{\ell_1,\ell_3} = \Big(\dfrac{\ell_1}{\ell_3}\Big)^3\eta.
	\end{equation*}
	The distinction we shall make is between $N_r^{(3\ell_1)} \geq \eta_{\ell_1,\ell_3}$ and $N_r^{(3\ell_1)} < \eta_{\ell_1,\ell_3}$. To keep the notation simple, we use the Iverson bracket $[.]$, defined by
	\begin{equation*}
		[A] = \left\{
		\begin{array}{ll}
			1 & \textmd{if $A$ is true,}\\
			0 & \textmd{otherwise.}
		\end{array}
		\right.
	\end{equation*}
	Namely, we write $[N_r^{(3\ell_1)} \geq \eta_{\ell_1,\ell_3}]$ instead of $\1(N_r^{(3\ell_1)} \geq \eta_{\ell_1,\ell_3})$. In summary, we bound
	\begin{multline}
		\label{eq:general_two_body_potential_bound_high_and_low_occupancy}
		\sum_{1\leq i< j\leq N}\1(\vert x_i - x_j\vert\leq \ell_2) \leq \sum_{r\in\ell_1\bbZ^3}[N_r^{(3\ell_1)} \geq \eta_{\ell_1,\ell_3}](N_r^{(3\ell_1)})^2\\
		+ \sum_{r\in\ell_1\bbZ^3}[N_r^{(3\ell_1)} < \eta_{\ell_1,\ell_3}]\sum_{1\leq i < j \leq N}\1(\vert x_i - x_j\vert \leq \ell_2)\1_{\Lambda_r^{(3\ell_1)}}(x_i)\1_{\Lambda_r^{(3\ell_1)}}(x_j).
	\end{multline}
	\\
	
	\noindent
	\textbf{Step 2: Boxes with many particles.} We start by bounding the first term in the right-hand side of \eqref{eq:general_two_body_potential_bound_high_and_low_occupancy} by
	\begin{equation*}
		\sum_{r \in \ell_1\bbZ^3}[N_r^{(3\ell_1)} \geq \eta_{\ell_1,\ell_3}](N_r^{(3\ell_1)})^2 \leq C\eta^{-1}\Big(\dfrac{\ell_1}{\ell_3}\Big)^{-3}\sum_{r \in \ell_1\bbZ^3}[N_r^{(3\ell_1)} \geq \eta_{\ell_1,\ell_3}](N_r^{(3\ell_1)})^3.
	\end{equation*}
	Next, we use \cite[Lemma 4.5]{TriayVis26} to deduce
	\begin{equation}
		\label{eq:convexity_bound_number_particles_box}
		\sum_{r \in \ell_1\bbZ^3}[N_r^{(3\ell_1)} \geq \eta_{\ell_1,\ell_3}](N_r^{(3\ell_1)})^3 \leq C\Big(\frac{\ell_1}{\ell_3}\Big)^6\sum_{s\in\tilde{\ell}_3\mathbb{Z}^3}[N_s^{(\tilde{\ell}_3)} \geq 3](N_s^{(\tilde{\ell}_3)})^3,
	\end{equation}
	where we defined $\tilde{\ell}_3 = \ell_3/\sqrt{6}$. The idea behind \eqref{eq:convexity_bound_number_particles_box} is explained after the present proof. Since $\tilde{\ell}_3$ was chosen such that any three $x_i,x_j,x_k\in\Lambda_s^{(\tilde{\ell}_3)}$ must satisfy $\vert (x_i - x_j,x_i - x_k)\vert \leq \ell_3$, the previous bound implies
	\begin{equation*}
		\sum_{r\in\ell_1\bbZ^3}[N_r^{(3\ell_1)} \geq \eta_{\ell_1,\ell_3}](N_r^{(3\ell_1)})^2 \leq C\eta^{-1}\Big(\dfrac{\ell_1}{\ell_3}\Big)^3\sum_{1\leq i < j < k\leq N}\1(\vert (x_i - x_j,x_i - x_k)\vert \leq \ell_3).
	\end{equation*}
	\\
	
	\noindent
	\textbf{Step 3: Boxes with few particles.} To bound the last term in \eqref{eq:general_two_body_potential_bound_high_and_low_occupancy}, we wish to use the Sobolev inequality. Applying it directly gives
	\begin{multline}
		\label{eq:sobolev_inequality_application}
		\sum_{r\in\ell_1\bbZ^3}[N_r^{(3\ell_1)} < \eta_{\ell_1,\ell_3}]\sum_{1\leq i < j \leq N}\1(\vert x_i - x_j\vert \leq \ell_2)\1_{\Lambda_r^{(3\ell_1)}}(x_i)\1_{\Lambda_r^{(3\ell_1)}}(x_j) \leq C\eta\ell_2^{2}\Big(\dfrac{\ell_1}{\ell_3}\Big)^3\sum_{i = 1}^N-\Delta_i\\
		+ \textmd{boundary term}.
	\end{multline}
	Although the factor in front of the kinetic energy has the desired form, the error that we get due to the boundary of the box $\Lambda_r^{(3\ell_1)}$ is too large, and we therefore need to be more subtle. To this end, we define the projections
	\begin{equation*}
		p_r^{(\ell)} = \dfrac{1}{\ell^3}\vert\1_{\Lambda_r^{(\ell)}}\rangle\langle\1_{\Lambda_r^{(\ell)}}\vert \quad \textmd{and} \quad q_r^{(\ell)} = \1_{\Lambda_r^{(\ell)}} - p_r^{(\ell)}.
	\end{equation*}
	Then, we use the resolution of the identity $\1_{\Lambda_r^{(3\ell_1)}} = p_r^{(3\ell_1)} + q_r^{(3\ell_1)}$ and the Cauchy--Schwarz inequality to deduce
	\begin{equation*}
		[N_r^{(3\ell_1)} < \eta_{\ell_1,\ell_3}]\sum_{1\leq i < j \leq N}\1(\vert x_i - x_j\vert \leq \ell_2)\1_{\Lambda_r^{(3\ell_1)}}(x_i)\1_{\Lambda_r^{(3\ell_1)}}(x_j) \leq \cJ_{1,r} + \cJ_{2,r}
	\end{equation*}
	with $\cJ_{1,r}$ and $\cJ_{2,r}$ given by
	\begin{equation*}
		\cJ_{1,r} = 2\sum_{1\leq i < j \leq N}q_{r,i}^{(3\ell_1)}\1(\vert x_i - x_j\vert \leq \ell_2)[N_r^{(3\ell_1)} < \eta_{\ell_1,\ell_3}]q_{r,i}^{(3\ell_1)}\1_{\Lambda_r^{(3\ell_1)}}(x_j)
	\end{equation*}
	and
	\begin{equation*}
		\cJ_{2,r} = 2\sum_{1\leq i < j \leq N}p_{r,i}^{(3\ell_1)}\1(\vert x_i - x_j\vert \leq \ell_2)[N_r^{(3\ell_1)} < \eta_{\ell_1,\ell_3}]p_{r,i}^{(3\ell_1)}\1_{\Lambda_r^{(3\ell_1)}}(x_j)
	\end{equation*}
	where $p_{r,i}^{(3\ell_1)}$ and $q_{r,i}^{(3\ell_1)}$ denote the projections acting on the $i$-th variable.
	\\
	
	\noindent
	\textit{Analysis of $\cJ_{1,r}$.} Define
	\begin{equation*}
		N_{r,i}^{(\ell)} = \sum_{\substack{j = 1\\ j\neq i}}^N\1_{\Lambda_r^{(\ell)}}(x_j).
	\end{equation*}
	Using that this operator commutes with $q_{r,i}^{(3\ell_1)}$, we get
	\begin{equation*}
		q_{r,i}^{(3\ell_1)}\1(\vert x_i - x_j\vert \leq \ell_2)[N_r^{(3\ell_1)} < \eta_{\ell_1,\ell_3}]q_{r,i}^{(3\ell_1)} \leq (q_{r,i}^{(3\ell_1)}\1(\vert x_i - x_j\vert \leq \ell_2)q_{r,i}^{(3\ell_1)})[N_{r,i}^{(3\ell_1)} < \eta_{\ell_1,\ell_3}].
	\end{equation*}
	Hence, an application of the Hölder inequality and of the Poincaré--Sobolev inequality (see e.g. \cite[Theorem 8.12]{LieLos01}) gives
	\begin{equation*}
		q_{r,i}^{(3\ell_1)}\1(\vert x_i - x_j\vert \leq \ell_2)[N_r^{(3\ell_1)} < \eta_{\ell_1,\ell_3}]q_{r,i}^{(3\ell_1)} \leq C\ell_2^2\nabla_i^*\1_{\Lambda_r^{(3\ell_1)}}(x_i)\nabla_i[N_{r,i}^{(3\ell_1)} < \eta_{\ell_1,\ell_3}],
	\end{equation*}
	which implies
	\begin{equation*}
		\cJ_{1,r} \leq C\eta\ell_2^2\Big(\dfrac{\ell_1}{\ell_3}\Big)^3\sum_{i = 1}^N\nabla_i^*\1_{\Lambda_r^{(3\ell_1)}}(x_i)\nabla_i.
	\end{equation*}
	Finally, summing over $r\in\ell_1\bbZ^3$ using that
	\begin{equation}
		\label{eq:box_recombiation}
		\sum_{r\in\ell_1\bbZ^3}\1_{\Lambda_r^{(3\ell_1)}}(x_i) = 27,
	\end{equation}
	we find
	\begin{equation*}
		\sum_{r\in\ell_1\bbZ^3}\cJ_{1,r} \leq C\eta\ell_2^2\Big(\dfrac{\ell_1}{\ell_3}\Big)^3\sum_{i = 1}^N-\Delta_i.
	\end{equation*}
	
	\noindent
	\textit{Analysis of $\cJ_{2,r}$.} Again, using that $N_{r,i}^{(3\ell_1)}$ commutes with $p_{r,i}^{(3\ell_1)}$, we get
	\begin{equation*}
		\cJ_{2,r} \leq C\eta\Big(\dfrac{\ell_1}{\ell_2}\Big)^3\Big(\dfrac{\ell_1}{\ell_3}\Big)^3N_r^{(3\ell_1)}.
	\end{equation*}
	Therefore, by summing over $r\in\ell_1\bbZ^3$ using \eqref{eq:box_recombiation}, we obtain
	\begin{equation*}
		\sum_{r\in\ell_1\bbZ^3}\cJ_{2,r} \leq C\eta\Big(\dfrac{\ell_1}{\ell_2}\Big)^3\Big(\dfrac{\ell_1}{\ell_3}\Big)^3N.
	\end{equation*}
	This concludes the proof of \eqref{eq:general_two_body_potential_bound_weaker} in Proposition~\ref{prop:general_two_body_potential_bound}.
\end{proof}

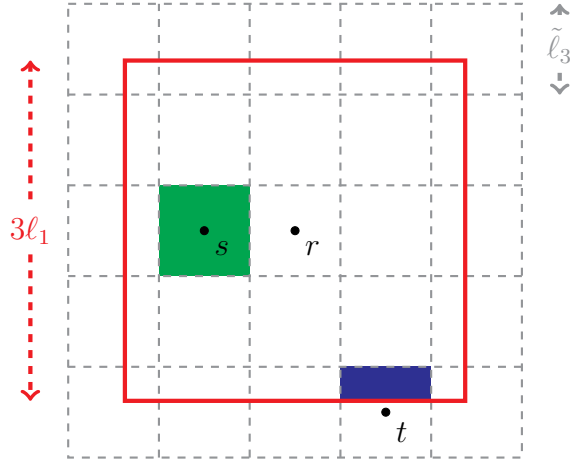
\begin{figure}[ht]
	\centering
	\begin{tikzpicture}
	\tikzmath{\w = 6; \wi = \w - \w/4; \se = \w/5;}
	
	\fill[Green] (-\w/2 + \se,-\w/2 + 2*\se) rectangle (-\w/2 + 2*\se,-\w/2 + 3*\se);
	
	\fill[Blue] (-\w/2 + 3*\se,-\wi/2) rectangle (-\w/2 + 4*\se,-\w/2 + \se);
	
	\draw[Gray, thick, dashed] (-\w/2,\w/2) -- (\w/2,\w/2) -- (\w/2,-\w/2) -- (-\w/2,-\w/2) -- cycle;
	
	\draw[Gray, thick, dashed] (\w/2,-\w/2 + \se) -- (-\w/2,-\w/2 + \se);
	\draw[Gray, thick, dashed] (\w/2,-\w/2 + 2*\se) -- (-\w/2,-\w/2 + 2*\se);
	\draw[Gray, thick, dashed] (\w/2,-\w/2 + 3*\se) -- (-\w/2,-\w/2 + 3*\se);
	\draw[Gray, thick, dashed] (\w/2,-\w/2 + 4*\se) -- (-\w/2,-\w/2 + 4*\se);
	\draw[Gray, thick, dashed] (-\w/2 + \se,-\w/2) -- (-\w/2 + \se,\w/2);
	\draw[Gray, thick, dashed] (-\w/2 + 2*\se,-\w/2) -- (-\w/2 + 2*\se,\w/2);
	\draw[Gray, thick, dashed] (-\w/2 + 3*\se,-\w/2) -- (-\w/2 + 3*\se,\w/2);
	\draw[Gray, thick, dashed] (-\w/2 + 4*\se,-\w/2) -- (-\w/2 + 4*\se,\w/2);
	
	\draw[Red, ultra thick] (-\wi/2,-\wi/2) -- (-\wi/2,\wi/2) -- (\wi/2,\wi/2) -- (\wi/2,-\wi/2) -- cycle;
	
	\filldraw[black] (0,0) circle (1.5pt) node[anchor=north west]{\large $r$};
	\filldraw[black] (-\se,0) circle (1.5pt) node[anchor=north west]{\large $s$};
	\filldraw[black] (\se,-2*\se) circle (1.5pt) node[anchor=north west]{\large $t$};
	\node[color=Red](a) at (-\w/2 - 0.5,0) {\large $3\ell_1$};
	\draw[color=Red, ultra thick, dashed,<-](-\w/2 - 0.5,\wi/2) -- (a);
	\draw[color=Red, ultra thick, dashed,->](a) -- (-\w/2 -0.5,-\wi/2);
	\node[color=Gray](b) at (\w/2 + 0.5,\wi/2 + 0.2) {\large $\tilde{\ell}_3$};
	\draw[color=Gray, ultra thick, dashed,<-](\w/2 + 0.5,\w/2) -- (b);
	\draw[color=Gray, ultra thick, dashed,->](b) -- (\w/2 + 0.5,\w/2 - \se);
\end{tikzpicture}
	\caption{Illustration of a coverage of the box ${\color{Red}\Lambda_r^{(3\ell_1)}}$ by the smaller boxes ${\color{Gray}\Lambda_s^{(\tilde{\ell}_3)}}$. Some of the {\color{gray}grey} boxes overlap only partially with the {\color{red}red} one. The quantities {\color{Green}$N_{r,s}^{(3\ell_1,\tilde{\ell}_3)}$} and {\color{Blue}$N_{r,t}^{(3\ell_1,\tilde{\ell}_3)}$} respectively count the number of particles in the {\color{Green}green} and {\color{Blue}blue} rectangles.}
	\label{fig:coverage_large_box_smaller_boxes}
\end{figure}

Let us explain the rough idea behind \eqref{eq:convexity_bound_number_particles_box}. For this, we need to define
\begin{equation*}
	N_{r,s}^{(\ell,\ell')} = \sum_{i = 1}^N\1_{\Lambda_r^{(\ell)}}(x_i)\1_{\Lambda_s^{(\ell')}}(x_i),
\end{equation*}
which counts the number of particles that are contained both in $\Lambda_r^{(\ell)}$ and $\Lambda_s^{(\ell')}$. Firstly, the condition
\begin{equation}
	\label{eq:number_particles_large_enough_condition}
	N_r^{(3\ell_1)} \geq \eta_{\ell_1,\ell_3} \geq C(\ell_1/\ell_3)^3
\end{equation}
states that the box $\Lambda_r^{(3\ell_1)}$ contains sufficiently many particles so that, if we consider any coverage of $\Lambda_r^{(3\ell_1)}$ by disjoint smaller boxes $\Lambda_s^{(\tilde{\ell}_3)}$ (see Figure \ref{fig:coverage_large_box_smaller_boxes}) and distribute the particles evenly across these smaller boxes, then each of them contains at least $3$ particles. Essentially, this ensures that, for any given configuration, there are few boxes $\Lambda_s^{(\tilde{\ell}_3)}$ that contain less than $3$ particles. This allows us to prove the bound
\begin{equation*}
	N_r^{(3\ell_1)} \leq C\sum_{s \in \tilde{\ell}_3\bbZ^3}[N_{r,s}^{(3\ell_1,\tilde{\ell}_3)} \geq 3]N_{r,s}^{(3\ell_1,\tilde{\ell}_3)}
\end{equation*}
under the condition \eqref{eq:number_particles_large_enough_condition}. Said differently, this means that we can discard the information about the small boxes that contain at most $2$ particles. By a convexity argument, this bound implies
\begin{equation*}
	(N_r^{(3\ell_1)})^3 \leq C\Big(\dfrac{\ell_1}{\ell_3}\Big)^6\sum_{s \in \tilde{\ell}_3\bbZ^3}[N_{r,s}^{(3\ell_1,\tilde{\ell}_3)} \geq 3](N_{r,s}^{(3\ell_1,\tilde{\ell}_3)})^3,
\end{equation*}
again under the condition \eqref{eq:number_particles_large_enough_condition}. Summing over $r \in \ell_1\bbZ^3$ gives \eqref{eq:convexity_bound_number_particles_box}.

\begin{proof}[Proof of \eqref{eq:two_body_potential_more_singuilar} from Proposition~\ref{prop:general_two_body_potential_bound}]
	The proof strategy is the same as that of \eqref{eq:general_two_body_potential_bound} and \eqref{eq:general_two_body_potential_bound_weaker}, but simpler.
	\\
	
	\noindent
	\textbf{Step 1: Localisation into boxes.} Following the exact same steps as in the proof of \eqref{eq:general_two_body_potential_bound_weaker}, we find
	\begin{equation}
		\label{eq:_two_body_potential_bound_high_and_low_occupancy}
		\sum_{1\leq i< j\leq N}\1(\vert x_i - x_j\vert\leq \ell_2) \leq \sum_{r\in\ell_1\bbZ^3}[N_r^{(3\ell_1)} \geq C_{\ell_1,\ell_3}](N_r^{(3\ell_1)})^2 + \sum_{r\in\ell_1\bbZ^3}[N_r^{(3\ell_1)} < C_{\ell_1,\ell_3}](N_r^{(3\ell_1)})^2,
	\end{equation}
	where we defined $C_{\ell_1,\ell_3}$ by
	\begin{equation*}
		C_{\ell_1,\ell_3} = C\Big(\dfrac{\ell_1}{\ell_3}\Big)^3,
	\end{equation*}
	for some large enough universal constant $C$.
	\\
	
	\noindent
	\textbf{Step 2: Boxes with many particles.} An application of \cite[Lemma 4.5]{TriayVis26} yields
	\begin{equation*}
		\sum_{r\in\ell_1\bbZ^3}[N_r^{(3\ell_1)} \geq C_{\ell_1,\ell_3}](N_r^{(3\ell_1)})^2 \leq C\Big(\dfrac{\ell_1}{\ell_3}\Big)^3\sum_{s\in\tilde{\ell}_3\bbZ^3}[N_s^{(\tilde{\ell}_3)} \geq 2](N_s^{(\tilde{\ell}_3)})^2,
	\end{equation*}
	where $\tilde{\ell}_3$ is chosen such that any two particles $x_i$ and $x_j$ belonging to the same box $\Lambda_s^{(\tilde{\ell}_3)}$ satisfy $\vert x_i - x_j\vert \leq \ell_3$. Therefore,
	\begin{equation*}
		\sum_{r\in\ell_1\bbZ^3}[N_r^{(3\ell_1)} \geq C_{\ell_1,\ell_3}](N_r^{(3\ell_1)})^2 \leq C\Big(\dfrac{\ell_1}{\ell_3}\Big)^3\sum_{1\leq i<j\leq N}\1(\vert x_i - x_j\vert \leq \ell_3).
	\end{equation*}
	\\
	
	\noindent
	\textbf{Step 3: Boxes with few particles.} Using that
	\begin{equation*}
		\sum_{r\in\ell_1\bbZ^3}\1_{\Lambda_r^{(3\ell_1)}}(x_i) = 27,
	\end{equation*}
	we get, by definition of $C_{\ell_1,\ell_3}$,
	\begin{equation*}
		\sum_{r\in\ell_1\bbZ^3}[N_r^{(3\ell_1)} < C_{\ell_1,\ell_3}](N_r^{(3\ell_1)})^2 \leq C\Big(\dfrac{\ell_1}{\ell_3}\Big)^3N.
	\end{equation*}
	This finishes the proof of \eqref{eq:two_body_potential_more_singuilar} from Proposition~\ref{prop:general_two_body_potential_bound}.
\end{proof}

\appendix

\section{Stabilisation by adding a soft-core: proof of Proposition~\ref{prop:stabilisation_soft_core}}

\label{app:stabilisation_soft_core}

The proof of Proposition~\ref{prop:stabilisation_soft_core} is from \cite[Lemma 21]{Tri18}; we recall it here for completeness. The idea is to first find a smooth function $\phi_1$ with sufficient decay at infinity (whose Fourier transform is therefore very well-behaved) that satisfies
\begin{equation*}
	-\phi_1 \leq V \quad \textmd{outside $B(0,R_0)$.}
\end{equation*}
Then, we use this function to create a bounded function $\phi$ that also satisfies
\begin{equation*}
	\phi \leq V \quad \textmd{outside $B(0,R_0)$,}
\end{equation*}
and whose Fourier transform is nonnegative (see Figure~\ref{fig:stabilisation_soft_core}). In $B(0,R_0)$, the function $\phi$ is bounded, so we can find a positive soft-core $C\1(\vert \cdot\vert \leq R_0)$ such that
\begin{equation*}
	\phi \leq V + C\1(\vert \cdot\vert \leq R_0) \quad \textmd{on $\bbR^3$.}
\end{equation*}

\begin{proof}[Proof of Proposition~\ref{prop:stabilisation_soft_core}]
	Recall that by assumption, $V$ is bounded from below and satisfies
	\begin{equation}
		\label{eq:stabilisation_decay_assumption}
		V \geq -\dfrac{C}{\vert\cdot\vert^{3 + \varepsilon}},
	\end{equation}
	for some $C,\varepsilon > 0$. Let $\chi$ be a smooth nonnegative function supported in the ball $B(0,R_0/2)$ such that $\int\chi = 1$. Define
	\begin{equation*}
		\phi_1 = C_0\dfrac{\1(\vert\cdot\vert \geq R_0/2)}{\vert\cdot\vert^{3 + \varepsilon}}\ast\chi,
	\end{equation*}
	for some fixed constant $C_0$ taken large enough. In particular, this constant is taken such that
	\begin{equation*}
		-\phi_1 \leq V \quad \textmd{outside $B(0,R_0)$;}
	\end{equation*}
	the fact that such a constant exists is easily verified using \eqref{eq:stabilisation_decay_assumption}. In addition, define
	\begin{equation*}
		\phi_2:x\in\bbR^3 \longmapsto \dfrac{\mu}{\vert x\vert}(e^{-\lambda\vert x\vert /2} - e^{-\lambda\vert x\vert}) \quad \textmd{and} \quad \phi = \phi_2 - 2\phi_1,
	\end{equation*}
	for some parameters $\mu > 0$ and $\lambda > 1$ to be determined later. An elementary computation gives
	\begin{equation*}
		\phi_2 \leq \phi_1 \quad \textmd{outside $B(0,R_0)$},
	\end{equation*}
	as long as $\mu\leq C_\varepsilon \lambda^{2+\varepsilon}$, which we henceforth assume. Thus, the potential $V$ satisfies
	\begin{equation*}
		\phi \leq V \quad \textmd{outside $B(0,R_0)$}.
	\end{equation*}
	Moreover, the function $\phi$ is easily seen to be bounded. Since $V$ is bounded from below, this implies
	\begin{equation*}
		\phi \leq V + C\1(\vert\cdot\vert \leq R_0),
	\end{equation*}
	for $C$ large enough.
	
	Finally, the claim follows from the fact that, for $c_\varepsilon\lambda^2 \leq \mu \leq C_\varepsilon \lambda^{2 + \varepsilon}$, the Fourier transform of $\phi$ is positive, which implies that $\phi$ is classically stable. Indeed, the function $\phi_1$ is smooth and decays sufficiently fast, so its Fourier transform satisfies
	\begin{equation*}
		\hat{\phi_1}(p) \leq \dfrac{C_\varepsilon}{(1 + \vert p\vert^2)^3},
	\end{equation*}
	and an elementary computation that relies on the Fourier transform of the Yukawa potential (see e.g. \cite[Theorem 6.23]{LieLos01}) gives
	\begin{equation*}
		\hat{\phi_2}(p) \geq \dfrac{C\lambda^2\mu}{(\lambda^2 + \vert p\vert^2)^2}.
	\end{equation*}
	This finishes the proof of Proposition~\ref{prop:stabilisation_soft_core}.
\end{proof}

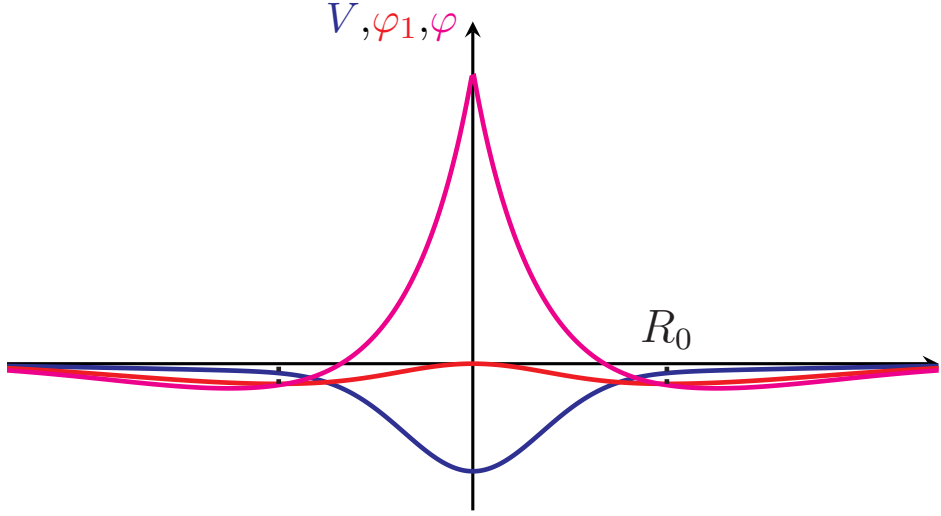
\begin{figure}[ht]
	\centering
	\resizebox{0.8\textwidth}{!}{%
		\begin{tikzpicture}
	\begin{axis}
		[
		axis lines=middle,
		axis line style={thick},
		y label style={anchor=east},
		ylabel={{\color{Blue}$V$},{\color{Red}$\phi_1$},{\color{Magenta}$\phi$}},
		xmin=-1.2, xmax=1.2,
		ymin=-1.5, ymax=3.5,
		xtick=\empty,
		ytick=\empty,
		samples=400,
		domain=-1:1,
		width=10cm,
		height=6cm,
		]
		
		\addplot[
		very thick,
		Blue,
		domain=-1.2:1.2,
		]
		{-exp(-(4*x)^2) -0.1*exp(-(x)^2)};
		
		\addplot[
		very thick,
		Red,
		domain=-1.2:1.2,
		]
		{-(1 - exp(-16*((x)^2)))*exp(-3*((x - 0.5)^2))/5 -(1 - exp(-16*((x)^2)))*exp(-3*((x + 0.5)^2))/5};
		
		\addplot[
		very thick,
		Magenta,
		domain=-1.2:1.2,
		]
		{0.6*(exp(-10*abs(x)/2) - exp(-10*abs(x)))/abs(x) -1.5*(1 - exp(-16*((x)^2)))*exp(-3*((x - 0.5)^2))/5 -1.5*(1 - exp(-16*((x)^2)))*exp(-3*((x + 0.5)^2))/5};
		
		\draw[Black,very thick,dotted] (0.5,{-(1 - exp(-16*((0.5)^2)))*exp(-3*((0.5 - 0.5)^2))/5 -(1 - exp(-16*((0.5)^2)))*exp(-3*((0.5 + 0.5)^2))/5}) -- (0.5,0);
		\draw[Black,very thick,dotted] (-0.5,{-(1 - exp(-16*((0.5)^2)))*exp(-3*((0.5 - 0.5)^2))/5 -(1 - exp(-16*((0.5)^2)))*exp(-3*((0.5 + 0.5)^2))/5}) -- (-0.5,0);
		
		\node[Black,anchor=south] at (0.5,0) {\large $R_0$};
	\end{axis}
\end{tikzpicture}
	}
	\caption{Illustration of the functions ${\color{Magenta}\phi_1}$ and ${\color{Red}\phi}$ for some given potential ${\color{Blue}V}$.}
	\label{fig:stabilisation_soft_core}
\end{figure}

\printbibliography

\end{document}